\documentclass[conference,letterpaper]{IEEEtran}
\usepackage[utf8]{inputenc} 
\usepackage[T1]{fontenc}
\usepackage{url}
\usepackage{ifthen}
\usepackage{cite}
\usepackage[cmex10]{amsmath} 

\usepackage{nidanfloat}

\usepackage{amsmath,amssymb,amsfonts}
\usepackage{algorithmic}
\usepackage{grffile}
\usepackage{graphicx}
\usepackage{epstopdf} 
\graphicspath{{.}{img/}{fig/}}
\DeclareGraphicsExtensions{.pdf,.png,.jpg}
\usepackage{textcomp}

\usepackage{subfigure}
\usepackage[subfigure]{tocloft}
\usepackage{flushend}

\usepackage{enumitem}
\usepackage{nameref}

\usepackage{proof}
\newtheorem{theorem}{Theorem}
\newtheorem{lemma}{Lemma}
\newtheorem{remark}{Remark}

\newcommand\F{\mathbb{F}}
\newcommand\R{\mathbb{R}}
\newcommand\E{\mathbb{E}}
\newcommand{\T}{\mathbf{T}}
\newcommand{\U}{\mathbf{U}}
\newcommand{\V}{\mathbf{V}}
\newcommand{\X}{\mathbf{X}}
\newcommand{\Y}{\mathbf{Y}}
\newcommand{\N}{\mathbf{N}}
\newcommand{\M}{\mathbf{M}}
\let\_\underline

\def\P{\mathbb{P}}
\DeclareMathOperator*{\argmax}{arg\,max}

\def\xsumint#1#2#3{{\setbox0=\hbox{${#2#3}{#1\int}$}
		\vcenter{\hbox{$#2#3$}}\kern-.5\wd0}}

\makeatletter
\def\@begintheorem#1#2{\@IEEEtmpitemindent\itemindent\relax\topsep 0pt\rmfamily\trivlist%
	\item[]\textit{\bfseries\noindent #1\ #2:} \itemindent\@IEEEtmpitemindent\relax\itshape}
\def\@opargbegintheorem#1#2#3{\@IEEEtmpitemindent\itemindent\relax\topsep 0pt\rmfamily \trivlist%
	\item[]\textit{\bfseries\noindent #1\ #2\ (#3):} \itemindent\@IEEEtmpitemindent\relax\itshape}
\def\@endtheorem{\endtrivlist}
\makeatother

\def\BibTeX{{\rm B\kern-.05em{\sc i\kern-.025em b}\kern-.08em
    T\kern-.1667em\lower.7ex\hbox{E}\kern-.125emX}}

\begin{document}

\title{Attacking Masked Cryptographic Implementations: Information-Theoretic Bounds}


\author{%
	\IEEEauthorblockN{Wei Cheng\IEEEauthorrefmark{1}, Yi Liu\IEEEauthorrefmark{1}, Sylvain Guilley\IEEEauthorrefmark{2}\IEEEauthorrefmark{1}, and Olivier Rioul\IEEEauthorrefmark{1}}\\[-2mm]
	\IEEEauthorblockA{\IEEEauthorrefmark{1}%
		LTCI, T\'el\'ecom Paris,
		Institut Polytechnique de Paris,
		91\,120, Palaiseau, France,
		firstname.lastname@telecom-paris.fr}
	\IEEEauthorblockA{\IEEEauthorrefmark{2}%
		Secure-IC S.A.S.,
		75\,014, Paris, France,
		sylvain.guilley@secure-ic.com
		\vspace{-0.3cm}
	}
}

\maketitle

\begin{abstract}
Measuring the information leakage is critical for evaluating the practical security of cryptographic devices against side-channel analysis. Information-theoretic measures can be used (along with Fano's inequality) to derive upper bounds on the success rate of any possible attack in terms of the number of side-channel measurements.  Equivalently, this gives lower bounds on the number of queries for a given success probability of attack. 
In this paper, we consider cryptographic implementations protected by (first-order) masking schemes, and derive several information-theoretic bounds on the efficiency of any (second-order) attack. 
The obtained bounds are generic in that they do not depend on a specific attack but only on the leakage and masking models, through the mutual information between side-channel measurements and the secret key.
Numerical evaluations confirm that our bounds reflect the practical performance of optimal maximum likelihood attacks.
\end{abstract}
\vspace{0.1cm}
\begin{IEEEkeywords}
Side-Channel Analysis, Information-Theoretic Metric, Masking Scheme, Success Rate, Monte-Carlo Simulation.
\end{IEEEkeywords}

\section{Introduction}

Since the seminal work by Kocher et al.\@~\cite{DBLP:conf/crypto/KocherJJ99}, side-channel analyses (SCAs) have been ones of the most powerful practical attacks against cryptographic devices. They exploit physically observable information leakage like instantaneous power consumption~\cite{DBLP:conf/crypto/KocherJJ99} or  electromagnetic radiation~\cite{Gandolfi:2001:EAC:648254.752700} to extract secret keys as illustrated in Fig.~\ref{fig:channel:side}. 

\begin{figure}[!h]
	\centering
	\vspace{-0.2cm}
	\includegraphics[width=0.4\textwidth]{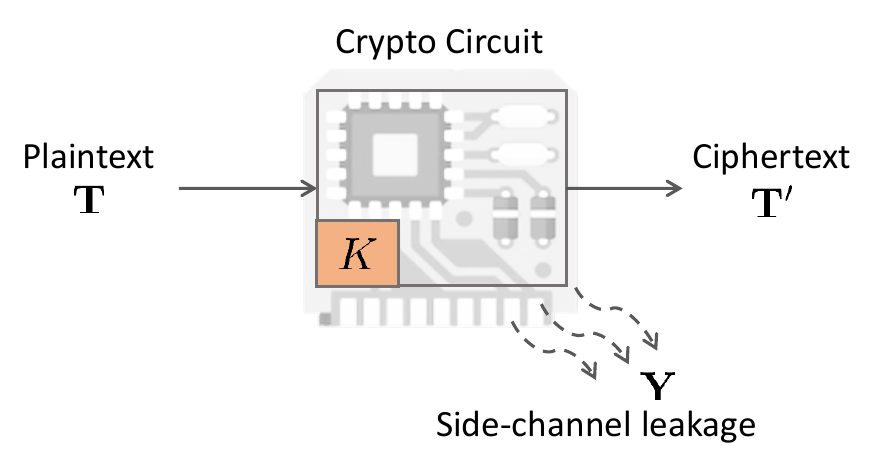}
	\vspace{-0.4cm}
	\caption{Side-channel in a nutshell. An adversary attempts to recover the secret key $K$ embedded in a cryptographic circuit by exploiting noisy side-channel leakage $\Y$ and public plaintext $\T$ (or ciphertext $\T'$).}
	\label{fig:channel:side}
	\vspace{-0.2cm}
\end{figure}

In last two decades, many different types of attacks have been proposed to exploit various types of leakages. In particular, Heuser et al.~\cite{DBLP:conf/ches/HeuserRG14} presented a channel representation of side-channel analysis to derive optimal (maximum likelihood) attacks that maximize success rate for a given leakage model. Other performance metrics such as guessing entropy also provide a fair comparison between different attacks~\cite{DBLP:conf/eurocrypt/StandaertMY09}.

To counteract SCAs, many countermeasures were proposed; \emph{masking} is a well-established protection which provides provable security~\cite{DBLP:conf/crypto/IshaiSW03,DBLP:conf/ches/RivainP10,DBLP:conf/eurocrypt/ProuffR13}. 
The idea is to split a sensitive
(secret-dependent) 
variable into several shares and perform computations separately on each (secret-independent) share. Since the masks themselves are leaking, sound attacks against masked implementations must be multidimensional and require an exponentially high number of measurements in the number of shares to succeed~\cite{DBLP:conf/eurocrypt/DucFS15}. 

A precise evaluation of the efficiency of \emph{any} possible side-channel attack in the presence of countermeasures is an open problem.  Given a set of side-channel measurements, can one establish a generic upper bound on the success rate of any attack?
Several bounds have been proposed in~\cite{DBLP:conf/eurocrypt/ProuffR13,DBLP:conf/eurocrypt/DucFS15} by approximations and inequalities. The resulting lower bounds (on the number of traces needed for a given success rate) are quite loose. Ch\'erisey et al.~\cite{DBLP:conf/isit/CheriseyGRP19,DBLP:journals/tches/CheriseyGRP19} derived several upper bounds on the success rate using mutual information, which are tight in assessing \emph{unprotected} cryptographic implementations. However, as we show in this paper, such bounds can also be very loose when targeting a \emph{protected} cryptographic implementation.

In this paper, we aim at providing tight bounds on the success rate of any SCA by leveraging information-theoretic tools. To do so, we consider a channel framework similar to the ones proposed in~\cite{DBLP:conf/ches/HeuserRG14,DBLP:conf/c2si/GuilleyHR17,DBLP:conf/isit/CheriseyGRP19,DBLP:journals/tches/CheriseyGRP19} but
enhance it for masking schemes.  
The overview of the framework is shown in Fig.~\ref{fig:channel:view} with notations introduced in the following Subsection.

\begin{figure*}[!t]
	\centering
	\vspace{-0.3cm}
	\includegraphics[width=0.85\textwidth,]{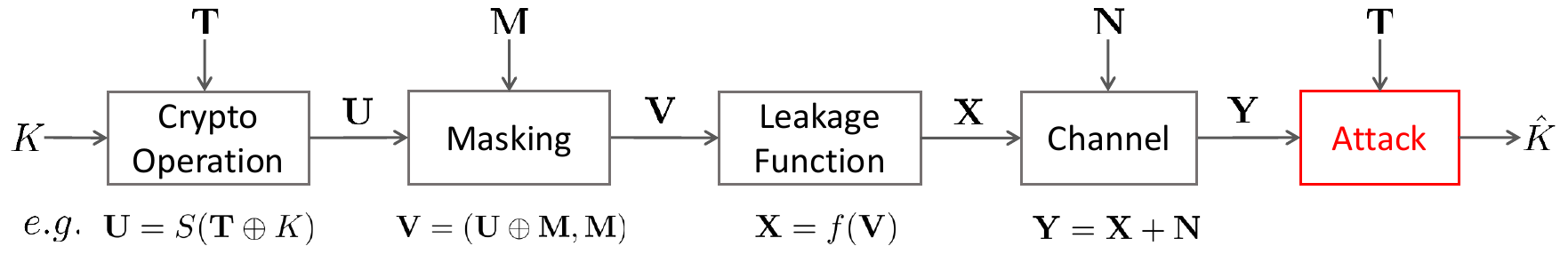}
	\vspace{-0.2cm}
	\caption{Channel representation of side-channel analysis of a masked cryptographic operation.}
	\label{fig:channel:view}
	\vspace{-0.4cm}
\end{figure*}

\subsection{Notations}

In the sequel, uppercase letters (e.g., $X$) denote random variables; 
lowercase letters (e.g., $x$) are for realizations (typically bytes); bold letters are for vectors, e.g., $\X=(X_1, X_2, \ldots, X_q)$. 
The cryptographic implementation typically works on bytes (e.g., of 8 bits) where the attacker, in a divide and conquer strategy, tries to recover each key byte~$K$ one by one. Let $T\oplus K$ be the bitwise exclusive or (XOR)   operation between a text byte and a key. For a sequence of $q$ text bytes $\T$ we write $\T\oplus K=(T_1\oplus K, T_2\oplus K, \ldots, T_q\oplus K)$.
Also let $w_H(X)$ denote the \emph{Hamming weight} of $X$ and $w_H(\X)=(w_H(X_1), w_H(X_2), \ldots, w_H(X_q))$. 

Throughout this paper we make the following notations as illustrated in Fig.~\ref{fig:channel:view}:
\begin{itemize}
	\item $K\in\F_{2^\ell}$ is the targeted key byte (typically $\ell=8$, e.g., for AES); 
	\item $\mathbf{T}\in\F_{2^\ell}^q$ denotes plaintext or ciphertext sequences, as vectors of length $q$; 
	\item $\mathbf{U}$ is the \emph{sensitive variable}, say $\mathbf{U} = S(\mathbf{T}\oplus K)$ where $S$ denotes a cryptographic operation like the Sbox in AES;
	\item 
	$\mathbf{V} = (\mathbf{U}\oplus \mathbf{M}, \mathbf{M})$ in a first-order Boolean masking with random mask $\mathbf{M}\in\F_{2^\ell}^q$; here $\V=(\V_1, \V_2)\in\F_{2^\ell}^{q\times 2}$ is a concatenation of $\V_1=\mathbf{U}\oplus \mathbf{M}$ and $\V_2=\mathbf{M}$; In the unprotected case (no masking) we would simply have $\mathbf{V} = \mathbf{U}$ as in~\cite{DBLP:conf/ches/HeuserRG14,DBLP:journals/tches/CheriseyGRP19};
	\item $\mathbf{X} = f(\mathbf{V}) = f(\mathbf{V}_1) + f(\mathbf{V}_2)$ is the so-called deterministic leakage, where e.g., $f=w_H$ in well-known Hamming weight model as in~\cite{DBLP:conf/ches/HeuserRG14}; more general models are possible;
	\item $\mathbf{Y} = \mathbf{X} + \mathbf{N}$ is the (noisy) leakage which models $q$ measurements (a.k.a. traces) in practice, where $\mathbf{N}$ is an independent i.i.d. noise (memoryless additive channel); in particular $\mathbf{N}\sim\mathcal{N}(0, \sigma^2\mathbf{I})$ for the AWGN channel.
		\item The attack is performed with a so-called distinguisher $\mathcal{D}$ which results in a guessed key $\hat{K}=\mathcal{D}(\Y,\T)$.
\end{itemize} 
From an information-theoretic perspective, it follows from Fig.~\ref{fig:channel:view} that 
conditionally on  $\mathbf{T}$, we have a Markov chain: 
$$K - \mathbf{U} - \mathbf{V} - \mathbf{X} - \mathbf{Y} - \hat{K}.$$ 

\begin{remark}
It is important to note that Fig.~\ref{fig:channel:view} is \emph{not} a genuine communication channel.
The designer wants the secret key $K$ to remain unknown and static (the same secret key is used for every side-channel use) as shown in Fig.~\ref{fig:channel:side}. Therefore, there is no message to be intentionally encoded and transmitted: $K$ leaks unintentionally. Besides, the actual (plain or encrypted) message $\mathbf{T}$ is public in our context and is supposedly known to the adversary. For all these reasons, our situation is totally different from problems such as those arising in a wiretap channel~\cite{DBLP:journals/corr/TyagiV14} for which a message is to be encoded, transmitted and decoded reliably in the presence of an eavesdropper.
\end{remark}

As recalled in~\cite{DBLP:journals/tches/CheriseyGRP19} for a memoryless channel, we have the following relation to single-letter quantities: $I(\X;\Y|\T)\leq q I(X;Y|T)$.  In particular, this explains why mutual information evaluation provides lower bounds on the number $q$ of queries as in~\cite[\S3.1]{DBLP:journals/tches/CheriseyGRP19}. Also, in~\cite[Theorem~4]{DBLP:journals/tifs/ChengGCMD21}, the leakage metric is:
$I(K; Y|T) = I(U; Y|T)$, which is implicitly connected to $q$~\cite{DBLP:journals/ccds/CarletG18}.


\subsection{Our Contributions}
In this work, we derive security bounds for side-channel attacks in the presence of first-order masking countermeasures. Instead of using theoretical upper bounds on mutual information (MI) $I(\X;\Y|\T)$ as in~\cite{DBLP:conf/isit/CheriseyGRP19,DBLP:journals/tches/CheriseyGRP19}, we numerically evaluate mutual information itself to derive bounds on the success rate thanks to Fano's inequality~\cite{information_theory}. 
We also use $I(\U;\Y|\T)$ in place of $I(\X;\Y|\T)$ in the presence of masking because the resulting bounds are much tighter. Numerical results in a commonly used side-channel setting will confirm that our new bound provides more accurate security guarantees for the chip designer in the context of masked cryptographic implementations.

The remainder of this paper is organized as follows. Section~\ref{sec:connect} provides connections between mutual informations (MIs) for different pairs of variables in a side-channel setting. Section~\ref{sec:bound:sr} presents several bounds on success rate. The numerical results for additive Gaussian noise are in Section~\ref{sec:numerical}. Finally, Section~\ref{sec:conclusion} concludes the paper.

\section{Theoretical Preliminaries}
\label{sec:connect}

\subsection{Links between MIs of Different Variables}

With the notations shown in Fig.~\ref{fig:channel:view} in the context of side-channel analysis, we have the following chain of equalities and inequalities for MIs on different pairs of variables.
\begin{lemma}
	\label{lemma:chain:mi}
	With the above definitions and notations, one has 
	\begin{equation}
	I(K; \Y|\T) = I(\U; \Y|\T) \leq I(\V; \Y|\T) = I(\X;\Y|\T).
	\end{equation}
\end{lemma}
As a result, we shall restrict ourselves only on the two MIs $I(\U; \Y|\T)$ and $I(\X; \Y|\T)$, where the former will necessarily give a better bound than the latter. 

\begin{proof}
Conditionally on $\T$, $K - \U - \Y$ is a Markov subchain; by the data processing inequality one has $I(K; \Y|\T) \leq I(\U; \Y|\T)$. Now since $\mathbf{U} = S(\mathbf{T}\oplus K)$ is a deterministic function of $K$ for fixed $\T$, $\U - K - \Y$ also forms a Markov chain conditionally on $\T$ and the converse inequality holds. This shows equality $I(K; \Y|\T) = I(\U; \Y|\T)$.
Similarly, conditionally on $\T$,
$\V - \X - \Y$ is a Markov subchain, but since $\X=f(\V)$, $\X - \V - \Y$ also forms a Markov chain. Then the data processing inequality in both directions implies equality $I(\V; \Y|\T) = I(\X;\Y|\T)$.
The data processing inequality applied to the Markov subchain $\U - \V - \Y$ gives
$I(\U; \Y|\T) \leq I(\V; \Y|\T)$ yet the converse is not true because of the presence of the unknown random mask $\M$.
\end{proof}

\begin{lemma}
With the above definitions and notations, for \emph{any} attack,
	\begin{equation}
	\label{lemma:ml:att}
	I(K; \hat{K})\leq I(K; \hat{K}|\mathbf{T}) \leq I(K;\mathbf{Y}|\mathbf{T}).
	\end{equation}
\end{lemma}
\begin{proof}
Since conditioning reduces entropy, $H(K|\hat{K})\geq H(K|\hat{K},\T)$. Then, since $K$ is independent of $\T$, we have $I(K; \hat{K}|\T) = H(K|\T) - H(K|\hat{K},\T) = H(K) - H(K|\hat{K},\T) \geq H(K) - H(K|\hat{K}) = I(K; \hat{K})$. This proves the first inequality.

Secondly, given $\T$, we have a Markov chain: $K-\Y-\hat{K}$, since for fixed~$\T$, $\hat{K}=\mathcal{D}(\Y,\T)$ is a deterministic function of $\Y$. The data processing inequality ends the proof.
\end{proof}

\begin{remark}
The ML (maximum likelihood) rule $\hat{k}=\mathcal{D}(\mathbf{y},\mathbf{t})=\argmax_k \P(\mathbf{Y}=\mathbf{y}|k, \mathbf{T} = \mathbf{t})$ gives the optimal distinguisher~\cite{DBLP:conf/ches/HeuserRG14} when it coincides with MAP (Maximum A Posterior) rule for uniformly distributed $K$ --- a common assumption in SCA.
\end{remark}

A trivial upper bound on $I(K;\mathbf{Y}|\mathbf{T})$ is as follows.
\begin{lemma}\label{lem:I_uyt:h_k}
With the above definitions and notations,
\begin{equation}
	\label{lemma:mi:h_k}
	I(K;\mathbf{Y}|\mathbf{T}) \leq H(K) \leq \ell.
\end{equation}
where typically $\ell=8$ bits.	
\end{lemma}
\begin{proof}
$I(K;\Y|\T)= H(K|\T) - H(K|\Y,\T) = H(K) - H(K|\Y,\T) \leq H(K)$.
\end{proof}
Lemma~\ref{lemma:mi:h_k} simply reflects the fact that the total amount of information any adversary could extract cannot exceed the information carried by the secret key, as measured by the entropy $H(K)$.
Notice that a common assumption in SCAs is that $K$ is uniformly distributed, in which case $H(K) = \ell$.

%


\subsection{Relation to Channel Capacity}

\begin{lemma}
With the above definitions and notations of Fig.~\ref{fig:channel:view},
	\begin{equation}
	I(\X;\Y)-I(\T;\Y) = I(\X;\Y|\T)\geq 0.
	\end{equation}
\end{lemma}

\begin{proof}
	Since $\T - \X - \Y$ forms a Markov chain, one has $H(\Y|\X,\T)=H(\Y|\X)$~\footnote{We use $H$ both discrete and continuous variables, even though $h$ is used more frequently for differential entropy of a continuous variable.}. Hence $I(\X;\Y|\T)=H(\Y|\T)-H(\Y|\X,\T)=H(\Y|\T)-H(\Y|\X)=H(\Y)-H(\Y|\X) - \bigl(H(\Y)-H(\Y|\T)\bigr)=I(\X;\Y)-I(\T;\Y)$.
\end{proof}
Note that the inequality $I(\X;\Y)-I(\T;\Y)\geq 0$ is also a direct consequence of the data processing inequality on the Markov chain  $\T-\X-\Y$.

\addtocounter{footnote}{1}
\footnotetext{We use $\log_2$ to have mutual information and entropy expressed in bits.}
\addtocounter{footnote}{-1}
One is led to define the \emph{capacity} of the side-channel (in bits per $q$ channel uses) as
\begin{equation}
\label{eqn:def:capa}
q\,C 
\!\!= \max_{\T-\X-\Y} I(\X;\Y|\T)
\!\!= \max_{\T-\X-\Y} I(\X;\Y)-I(\T;\Y),
\end{equation}
where the maximum is taken over all distributions of $\X$ given $\T$ such that $\T-\X-\Y$ is a Markov chain.
Because the ``side information'' $\T$ is known both at the ``encoder'' (leaking crypto) and ``decoder'' (attack), the capacity can be determined in the usual way:
\begin{lemma}\label{lemm-capacity}
With the above definitions and notations of Fig.~\ref{fig:channel:view} where the side-channel is independent of $\T$, one has
	\begin{align}
	q\,C=\max_{\X} I(\X;\Y)
	\end{align}
	where the maximum is taken over all channel input distributions $\X$.
\end{lemma}


\begin{proof}
Since $I(\X;\Y|\T)=\E_{\T} I(\X;\Y|\T=\mathbf{t})$, we can choose $p(\mathbf{x}|\mathbf{t})=p(\mathbf{x})$ to maximize each $I(\X;\Y|\T=\mathbf{t})$ to achieve channel capacity  in~\eqref{eqn:def:capa}. As the optimal distribution does not depend on $\mathbf{t}$, it also maximizes the expectation $\E_{\T} I(\X;\Y|\T=\mathbf{t})=I(\X;\Y|\T)$ and thus $\max_{\T-\X-\Y} I(\X;\Y|\T)=\max_{\X} I(\X;\Y)$.
\end{proof}

\begin{remark}
This result is also obtained by taking $\X$ (and thus $\Y$) independent of $\T$ such that $I(\T;\Y)=0$ in the preceding Lemma.
We could also consider the more general situation where the channel also depends on $\T$. In this case we would have $C=\E\{C_T\}$ where $q\, C_\mathbf{t}=\max_{\X} I(\X;\Y|\T=\mathbf{t})$.
\end{remark}

\begin{remark}
As it turns out, capacity yields an upper bound on $I(K;\mathbf{Y}|\mathbf{T})$ which can improve the trivial upper bound of Lemma~\ref{lem:I_uyt:h_k}. This does not mean, however, that one is faced with a channel coding problem since the ``encoder'' hence $X$'s distribution cannot be chosen by the attacker.
\end{remark}

\section{Bounds on the Success Probability of Attack}
\label{sec:bound:sr}


By combining Lemmas~\ref{lemma:chain:mi},~\ref{lemma:ml:att} and~\ref{lemma:mi:h_k}, we have $I(K; \hat{K})\leq I(\U; \Y|\T)\leq H(K)$. Now the probability of success (estimated as the \emph{success rate} in SCA) is defined as:
$P_s = \P(\hat{K}=K)$.
The corresponding ``probability of error'' (of attack failure) is $P_e = 1-P_s$. 
Using Fano's inequality~\cite{information_theory} we end up with the following theorem.
\begin{theorem}
\label{theorem:sr:bound}
Given the side-channel setting as in Fig.~\ref{fig:channel:view}, we have
\begin{equation}
d_P(P_s) \leq I(\U; \Y|\T),
\end{equation}
where $d_P(p) = H(K) - H_2(p) - (1-p)\log(2^\ell -1)$ and $H_2(p) = -p\log p - (1-p) \log(1-p)$, for $p\in[2^{-\ell}, 1]$. (Recall that $\ell$ denotes the number of bits in $K=k$.)
\end{theorem}
\begin{proof}
By Fano's inequality~\cite{information_theory} and Lemma~\ref{lemma:ml:att}, we have $H(K) - H_2(P_s) - (1-P_s)\log(2^\ell -1)\leq H(K)-H(K|\hat{K})= I(K; \hat{K}) \leq I(\U; \Y|\T)$. 
\end{proof}

Since $d_P(p)$ is strictly increasing for $p\in[2^{-\ell}, 1]$~\cite[\S A]{cryptoeprint:2019:491}
and $I(\U; \Y|\T)$ increases as $q$ increases, 
Theorem~\ref{theorem:sr:bound} not only provides an upper bound on $P_s$, but also gives a \emph{lower bound} on the number of queries $q$ to obtain a specific value of $P_s$.

\medskip

\begin{remark}
\label{remark:mi:xyt}
A much looser bound on $P_s$ can obtained from Lemmas~\ref{lemma:chain:mi} and~\ref{lemm-capacity}. Using Theorem~\ref{theorem:sr:bound}, one readily obtains
\begin{equation}\label{eq:bound:I_xyt}
d_P(P_s)\leq I(\X; \Y|\T) \leq q\, C
\end{equation}
where $C$ is the side-channel capacity, which is $C=\frac{1}{2}\log (1+\text{SNR})$ for an AWGN channel\,\footnotemark. 

However, as we will show below, this bound is useless in evaluating masked implementations, particularly because $I(\X; \Y|\T)$ is unbounded (compare with Lemma~\ref{lemma:mi:h_k}).
In fact, we will show in next section that $I(\X; \Y|\T)$ is very close to the capacity $q\, C$ in the presence of a Boolean masking on an AWGN channel, hence it increases linearly in $q$ without bound.
\end{remark}

%


\section{Application to Hamming Weight Leakages with Additive White Gaussian Noise}
\label{sec:numerical}

By the equalities of Lemma~\ref{lemma:chain:mi}, the only two MIs that need to be evaluated are $I(\X; \Y|\T)$ and $I(\U; \Y|\T)$.
Taking notations from Fig.~\ref{fig:channel:view}, we calculate both MIs numerically. We have
\begin{align}
\begin{split}
I(\X; \Y|\T) &= H(\Y|\T) - H(\Y|\X, \T), \\
I(\U; \Y|\T) &= H(\Y|\T) - H(\Y|\U, \T),
\label{eqn-IXYT}
\end{split}
\end{align}
where for the AWGN channel
\begin{equation}
H(\Y|\X,\T) = H(\Y|\X) = H(\N) 
= \frac{q}{2}\log\left( 2\pi e\sigma^2 \right),
\end{equation}
and where $H(\Y|\T)$ and $H(\Y|\U, \T) = H(\Y|\U)$ are estimated by Monte-Carlo simulations as shown next.

\subsection{
Monte-Carlo Simulation}

As the number of traces $q$ gets very large, direct integration to evaluate mutual information becomes infeasible. 
Monte-Carlo simulation is a well-known method to estimate expectations of a function under certain distribution by repeated random sampling.
We can then estimate the first term $H(\Y|\T)$ in~\eqref{eqn-IXYT} by randomly drawing $N_C$ samples:
\begin{align}
\begin{split}
H(\Y|\T) &= \int_{\mathbf{y}}\sum_{\mathbf{t}} p(\mathbf{y}, \mathbf{t}) \log\frac{1}{p(\mathbf{y}|\mathbf{t})} \;\text{d}\mathbf{y} \\[-1ex]
&= \lim_{N_C\rightarrow \infty} - \frac1{N_C} \sum_{j=1}^{N_C} \log p(\mathbf{y}^j|\mathbf{t}^j),
\label{eqn:hyt:est}
\end{split}
\end{align}
where each $(\mathbf{t}^j, \mathbf{y}^j)$, for $1\leq j\leq N_C$, is drawn randomly. The estimation in \eqref{eqn:hyt:est} is sound based on the law of large numbers~\cite[Chap. 3]{information_theory} and it has been numerically verified in~\cite{DBLP:journals/tches/CheriseyGRP19}.
Similarly, $H(\Y|\U)$ can be estimated using Monte-Carlo simulation by $H(\Y|\U)= - \frac1{N_C} \sum_{j=1}^{N_C} \log p(\mathbf{y}^j|\mathbf{u}^j)$.

\begin{figure}[!h]
	\centering
	\vspace{-0.2cm}
	\includegraphics[width=0.45\textwidth,height=4cm]{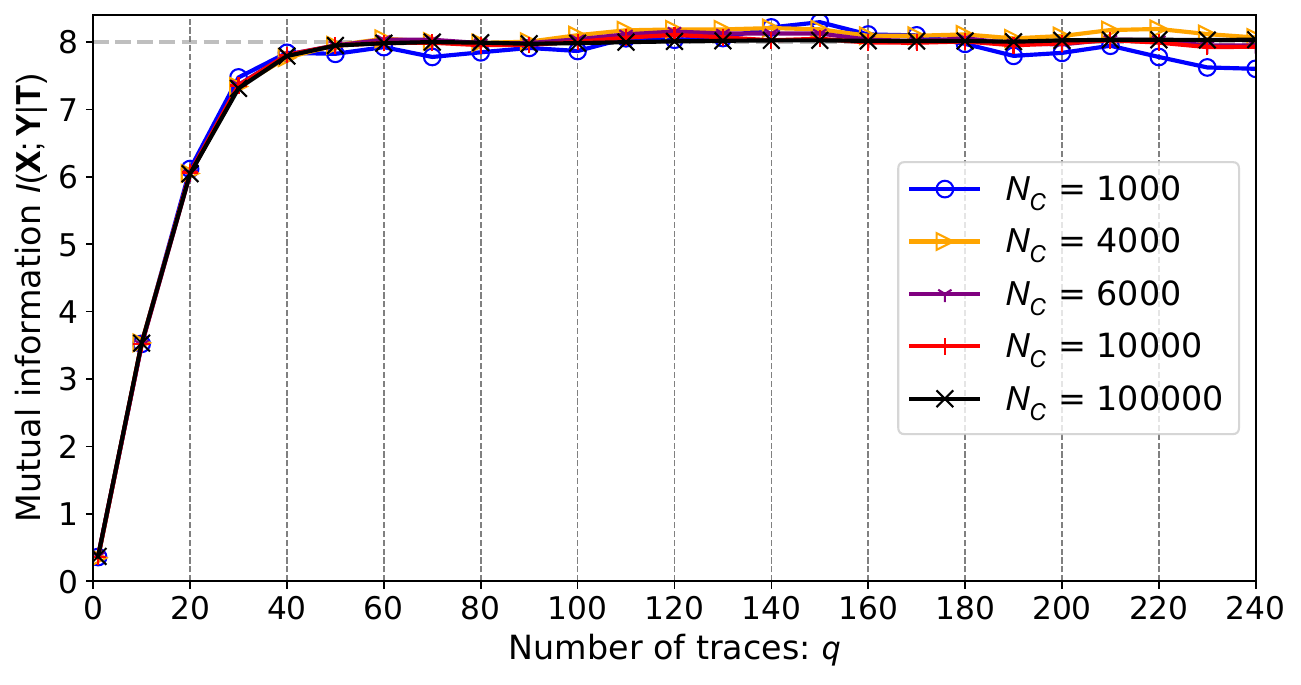}
	\vspace{-0.3cm}
	\caption{Monte-Carlo simulation with various $N_C$ draws where $\sigma^2=10.00$.}
	\vspace{-0.0cm}
	\label{fig:mi:I_xy_t:conv}
\end{figure}

The accuracy of Monte-Carlo simulation highly depends on the number of samples. 
As an illustration, consider the unprotected case where there is no masking and for which $I(\X; \Y|\T)$ is bounded by $H(K)=8$ bits.
As shown in Fig.~\ref{fig:mi:I_xy_t:conv}, the estimation of $I(\X; \Y|\T)$
gets more accurate by using larger $N_C$. In particular, this estimation on $I(\X; \Y|\T)$ is accurate enough by using only $N_C=100,000$ draws. For all results in this paper we use $N_C=1,000,000$ 
to obtain a very stable estimation.

\subsection{Numerical Results for First-order Boolean Masking}

Here $(\mathbf{t}^j, \mathbf{y}^j)$, for $1\leq j\leq N_C$, is drawn i.i.d.\@
according to this process:
\begin{itemize}
\item $\mathbf{t}^j\sim\mathcal{U}(\F_{2^\ell}^q)$,
\item $\mathbf{m}^j\sim\mathcal{U}(\F_{2^\ell}^q)$,
\item $k^j\sim\mathcal{U}(\F_{2^\ell})$, and
\item $\mathbf{y}^j\sim\mathcal{N}(w_H(S(\mathbf{t}^j\oplus {k}^j) \oplus\mathbf{m}^j)+w_H(\mathbf{m}^j), \sigma^2 \mathbf{I}_q)\in\R^q$.
\end{itemize}
Note that we consider the zero-offset leakage~\cite{DBLP:journals/ccds/CarletG18} where the leakages of each share are summed together (see the sum of two Hamming weights above).
%
For each draw $(\mathbf{t}, \mathbf{y})$, we have
\begin{align}
\begin{split}
p(\mathbf{y}|\mathbf{t}) &= \sum_k p(k) p(\mathbf{y}|\mathbf{t}, k) 
= \sum_k p(k) \prod_{i=1}^{q} p(\mathbf{y}_i|\mathbf{t}_i, k)  \\
&= \sum_k p(k) \prod_{i=1}^{q} \sum_{m_i} p(m_i) p(\mathbf{y}_i|\mathbf{t}_i, k, m_i)  \\
&= \sum_k p(k) \prod_{i=1}^{q}\sum_{m_i}p(m_i)\frac{e^{\frac{-\left( \mathbf{y}_i - f(\mathbf{t}_i, k, m_i) \right)^2}{2\sigma^2}}}{({2\pi \sigma^2})^{1/2}},
\label{eqn:pyt:mask}
\end{split}
\end{align}
where $f(\mathbf{t}_i, k, m_i) = w_H(S(\mathbf{t}_i\oplus k)\oplus m_i) + w_H(m_i)$ is the zero-offset leakage under Hamming weight model.
Again, taking $K\in\F_{2^\ell}$ uniformly, and considering that all masks are i.i.d.\@ $\sim\mathcal{U}(\F_{2^\ell})$, we have
\begin{multline}
\log p(\mathbf{y}|\mathbf{t}) = -\ell(q+1) - \frac{q}{2} \log\left( 2\pi\sigma^2 \right)\\
+ \log \sum_k \prod_{i=1}^{q} \sum_{m} e^{\frac{-\left( \mathbf{y}_i - f(\mathbf{t}_i, k, m) \right)^2}{2\sigma^2}}.
\label{eqn:pyt:mask:log}
\end{multline}
\vspace*{-2ex}
\begin{multline}
\log p(\mathbf{y}|\mathbf{u}) = -q\ell - \frac{q}{2} \log\left( 2\pi\sigma^2 \right)\\
+ \log \prod_{i=1}^{q} \sum_{m} e^{\frac{-\left( \mathbf{y}_i - f'(\mathbf{u}_i, m) \right)^2}{2\sigma^2}}.
\label{eqn:pyu:mask:log}
\end{multline}
where $f'(\mathbf{u}_i, m)=w_H(\mathbf{u}_i\oplus m) + w_H(m)$.


\begin{figure*}[!b]
	\vspace{-0.4cm}
	\centering
	\subfigure[$I(\X; \Y|\T)$]{
	\includegraphics[width=0.49\textwidth]{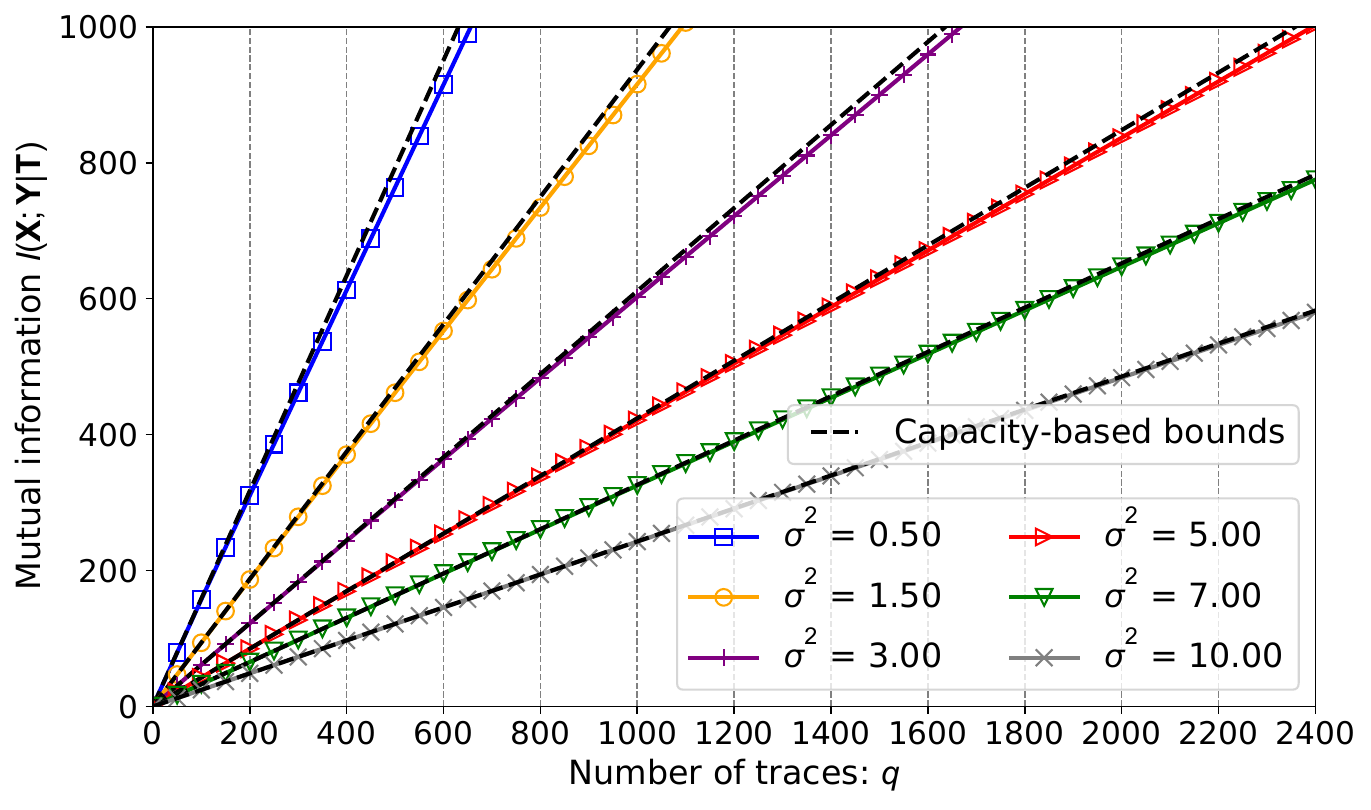}
	\label{fig:mi:mask:1d:I_xyt}
	}\hspace{0.0cm}
	\subfigure[$I(\U; \Y|\T)$]{
		\includegraphics[width=0.47\textwidth]{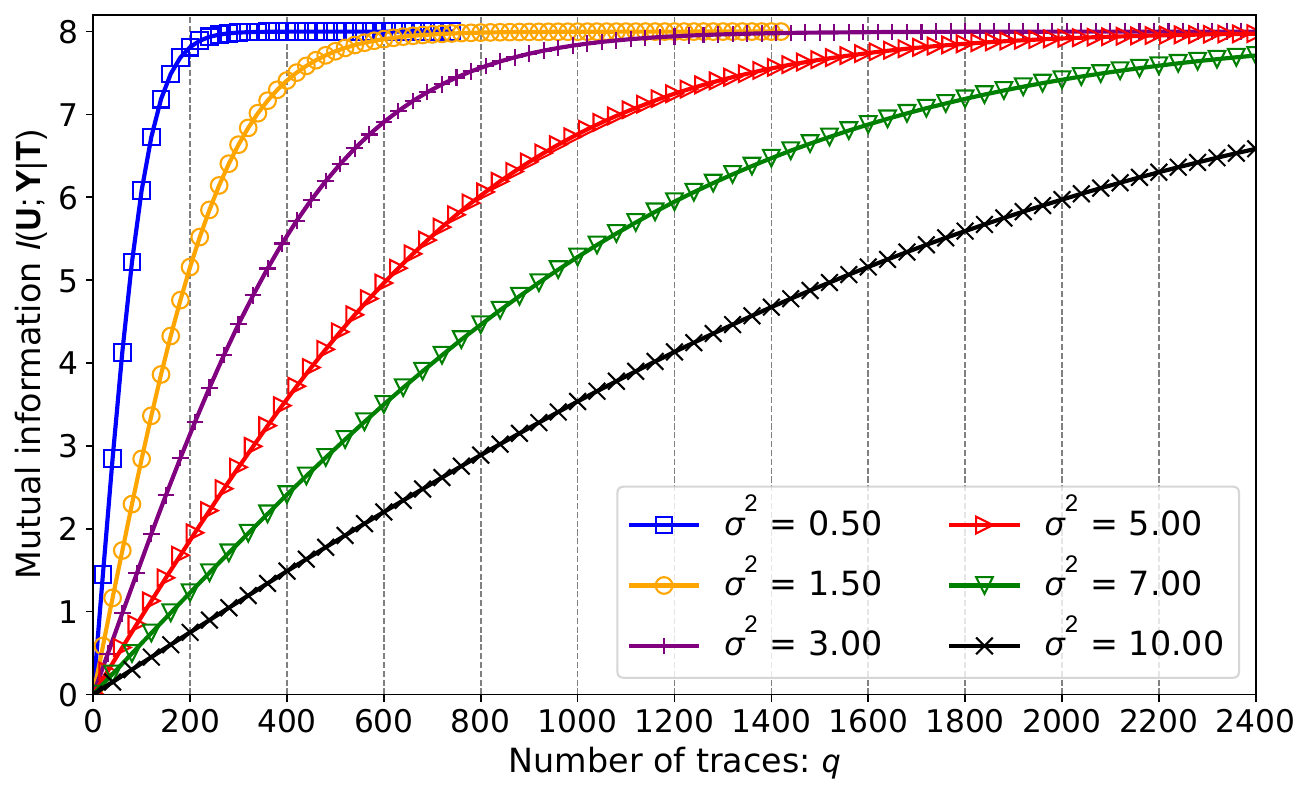}
	\label{fig:mi:mask:1d:I_uyt}
	}
	\vspace{-0.2cm}
	\caption{Evolution of mutual information $I(\X; \Y|\T)$ and $I(\U; \Y|\T)$ with the number of traces under different levels of noise in masked cases, with $N_C = 1,000,000$. Note that $I(\X; \Y|\T)$ is upper bounded by Shannon’s channel capacity, while $I(\U; \Y|\T)$ is upper bounded by $H(K)=8$ bits.}
	\label{fig:mi:mask:1d:xyt:uyt}
	\vspace{-0.3cm}
\end{figure*}



The numerical results of $I(\X; \Y|\T)$ are depicted in Fig.~\ref{fig:mi:mask:1d:I_xyt}.
It clearly appears that the effect of masking is to increase the values of $I(\X;\Y|\T)$ without bound.
This motivates our focus on $I(\mathbf{U};\Y|\T)$. 
The dotted black lines in Fig.~\ref{fig:mi:mask:1d:I_xyt} show that upper bounds given by~\eqref{eq:bound:I_xyt} are very tight.


As shown in Fig.~\ref{fig:mi:mask:1d:I_uyt}, $I(\U; \Y|\T)$ is bounded as expected by $H(K)$ in  Lemma~\ref{lem:I_uyt:h_k}. Particularly, given the same noise level, the number of traces needed to obtain $I(K; \Y|\T) = I(\U; \Y|\T) \approx 8~\text{bits}$ is much larger than in the unprotected case. 
The curves $I(\U; \Y|\T)$ vs $\sigma^2$ also look homothetic with a scale of $\sigma^2$. 
%
This is justified by a simple scaling argument:
if the number of traces for a given set of $(\T,\U)$ is doubled,
then the mutual information is the same as with the nominal number of queries, but with SNR doubled as well.

\subsection{Bounds on Success Rate in Masked Implementations}

By Theorem~\ref{theorem:sr:bound}, we have an upper bound on probability of success $P_s$. This equivalently gives a lower bound on the minimum of $q$ to get a specific $P_s$. 

Numerical results are shown in Fig.~\ref{fig:mi:mask:1d:sr} where we present several instances with different levels of Gaussian noises. In particular, the ML attacks utilize the higher-order distinguishers which have been demonstrated to be optimal in the presence of masking~\cite{DBLP:conf/asiacrypt/BruneauGHR14}. In order to evaluate $P_s$ of ML attacks, each attack is repeated $200$ times to have a more accurate success rate.

\begin{figure*}[!htbp]
	\centering
	\vspace{-0.2cm}
	\includegraphics[width=0.99\textwidth,height=7.9cm]{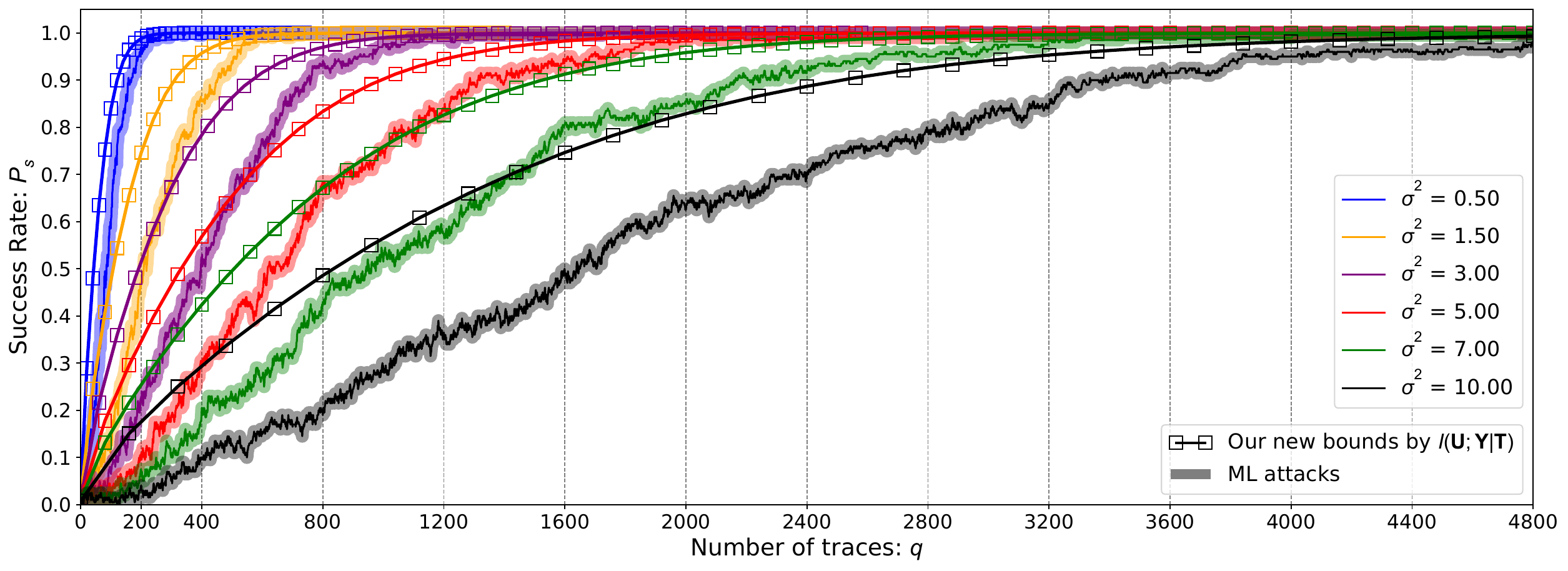}
	\vspace{-0.3cm}
	\caption{Application and comparison of bounds on success rate. We present six instances with different noise levels by using $q_{\max}=4800$ traces. Note that we omit the bounds given by $I(\X;\Y|\T)$ as they are invisible when plotted together with bounds given by $I(\U;\Y|\T)$.}
	\label{fig:mi:mask:1d:sr}
	\vspace{-0.3cm}
\end{figure*}


Figure~\ref{fig:mi:mask:1d:sr} already shows the usefulness of the bound given by $I(\U;\Y|\T)$. 
Indeed, a commonly used metric on attacks is the minimum number of traces to reach $P_s\geq 95\%$. Considering $\sigma^2=3.00$ in Fig.~\ref{fig:mi:mask:1d:sr}, we set $P_s=95\%$ and the ML attack needs around $q=800$ traces, 
where our new bound gives $q=720$, while the bound proposed in~\cite{DBLP:journals/tches/CheriseyGRP19} by using $I(\X;\Y|\T)$ only gives $q=12$. The comparison would be even worse for higher levels of noise.

Figure~\ref{fig:mi:mask:sr:bound} provides a more detailed comparison
by plotting the predicted minimum numbers of traces $q_{\min}$ reaching $P_s\geq 95\%$ given by both $I(\U;\Y|\T)$ and $I(\X;\Y|\T)$.
These curves show that our new bound is much tighter than the previous one from the state-of-the-art~\cite{DBLP:conf/isit/CheriseyGRP19,DBLP:journals/tches/CheriseyGRP19},
as it captures the masking scheme --- recall from Fig.~\ref{fig:channel:view} that the masking countermeasure step is between $\U$ and $\Y$ but not between $\X$ and $\Y$.

\begin{figure}[!tbp]
	\centering
	\includegraphics[width=0.49\textwidth]{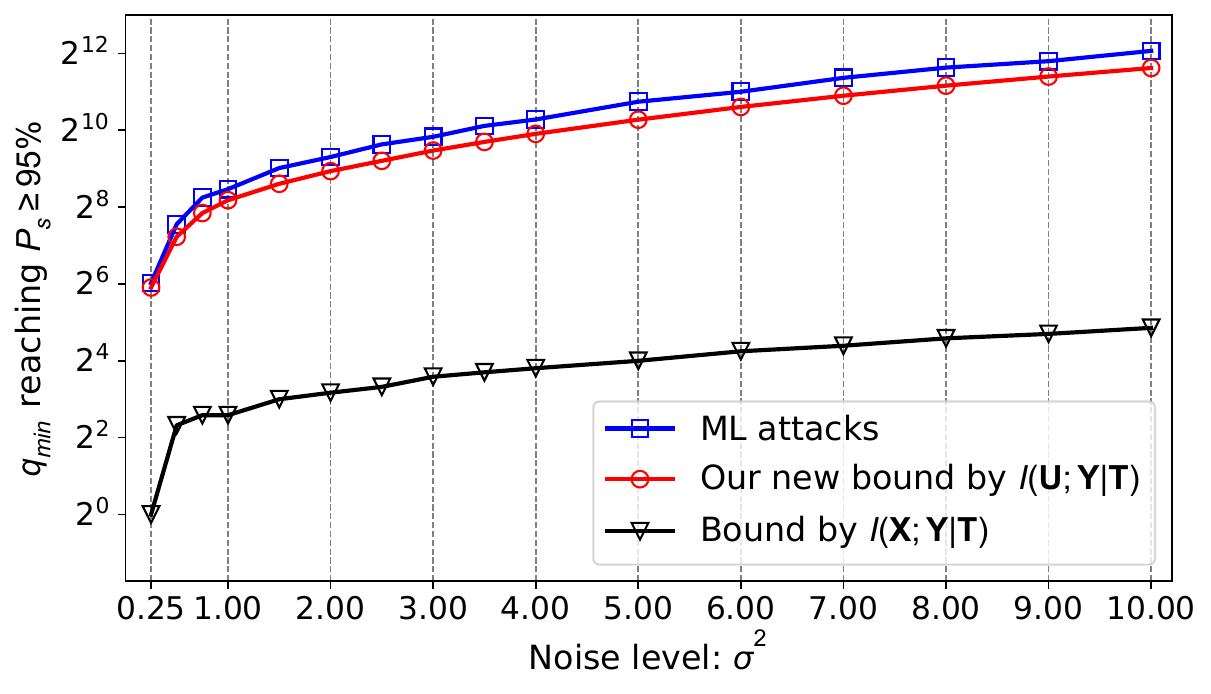}
	\vspace{-0.74cm}
	\caption{Comparison of the minimum number of traces $q_{\min}$ to reach $P_s\geq 95\%$ predicted by our new bound, by $I(\X;\Y|\T)$ as in~\cite{DBLP:journals/tches/CheriseyGRP19} and also the baseline given by an ML attack.}
	\vspace{-0.5cm}
	\label{fig:mi:mask:sr:bound}
\end{figure}


\section{Conclusions}
\label{sec:conclusion}

We derived security bounds for side-channel attacks in the presence of countermeasures (first-order masking).
To do this, we leveraged the seminal framework from Ch\'erisey et al.~\cite{DBLP:conf/isit/CheriseyGRP19,DBLP:journals/tches/CheriseyGRP19} and extended it to the masking case of a protection aiming at randomizing the leakage.

The generalization not only enhances bounds compared to Ch\'erisey et al., but also improves on the computation method for the security metric, by resorting to a powerful information estimation based on the Monte Carlo method.
Our results provide quantitative bounds allowing for the theoretical (``pre-silicon'') evaluation of protections applied on top of a given cryptographic algorithm in designing secure circuits. As a perspective, we will push forward the practical applications of our findings in evaluating concrete security level of cryptographic circuits.



\flushend


\providecommand{\noopsort}[1]{}

\end{document}